\documentclass[11pt]{article}

\usepackage{latexsym}
\usepackage{color,graphicx}
\usepackage{amssymb}
\usepackage{amsmath}
\usepackage{amsthm}
\usepackage{txfonts}
\usepackage{enumitem}
\usepackage{fullpage}
\usepackage{comment}

\newtheorem{theorem}{Theorem}[section]
\newtheorem{lemma}[theorem]{Lemma}

\newcommand{\cM}{\mathcal{M}}

\newcommand{\Mshared}[1][]{\cM \if!#1!\else (#1) \fi}

\title{On scheduling coupled tasks with exact delays to minimize maximum  lateness}
\author{
Wies\l{}aw Kubiak\\
\\
\small{\emph{Memorial University}}\\
\small{\emph{St. John's, Canada}}
}
\begin{document}
\maketitle
\begin{abstract} This paper studies scheduling coupled tasks with exact delays to minimize maximum lateness. The first task has processing time $p>0$ and the second $b_i\geq 0$, also the second needs to start exactly $p$ units of time after the completion of the first. The couple has due date $d_i$. The tasks are scheduled on a single machine to minimize maximum lateness. The problem has been left open in the literature which offer hardly any results on scheduling coupled tasks with exact delays to minimize maximum lateness. The paper shows polynomial time algorithms  for \emph{agreeable} ($d_i\geq d_j$ implies $b_i\geq b_j$) and \emph{disagreeable} ($d_i\geq d_j$ implies $b_i\leq b_j$) cases. The complexity of the general problem remains open.
\end{abstract}


\section{Introduction} 

We refer the reader to the recent comprehensive reviews of the literature  on coupled tasks scheduling and applications by Chen and Zhang \cite{CZ20}, and Khatami et al \cite{KSC20}. This paper follows the notation used 
%Chen and Zhang \cite{CZ20} and Khatami et al \cite{KSC20}. 
in those two papers with one exception, to avoid confusion we reserve the notation $L_j$ for the lateness $L_j=C_j-d_j$ of job $J_j$ in this paper. This is the standard notation used in the scheduling literature,  see Blazewicz et al \cite{BEPSW07} for instance,  and we let $l_j$   denote the exact delay between the tasks of job $J_j$. 

Majority of literature on the coupled tasks scheduling has focused of the minimization of makespan thus far. Only few papers have dealt with other objectives, mostly total completion time. For the minimization of total completion time the problem
$1|(p, p, b_j)|\sum C_j$ was shown polynomial by  Chen and Zhang \cite{CZ20}, yet the problem $1|(1, l_j, 1)|\sum C_j$ was shown NP-hard in the strong sense by Kubiak \cite{K22}. Recall that the triplet $(a_j,l_j,b_j)$ in the middle field of the traditional three-field scheduling notation of Graham et. al. \cite{GLLRK79} indicates that the first task of a job $j$ has processing time $a_j$, the exact delay between tasks of job $j$ equals  to $l_j$, and the processing time of the second task of job $j$ has processing time $b_j$. 
Thus the triplet $(p,p,b_j)$ denotes the class of instances having the processing times of first tasks as well as the delays all fixed to $p$, which is a part of the input, and processing times $b_j$ of second tasks being arbitrary and job-dependent;   
the triplet $(1,l_j,1)$ denotes the class of instances having the processing times of all tasks equal to $1$, and the delays $l_j$ being arbitrary and job-dependent.

To our knowledge no papers have dealt with the minimization of maximum lateness problems unless their complexity immediately follows from the complexity of the minimization of makespan of their counterpart problems.  In particular the maximum lateness minimization $1|(p, p, b_j)|\max L_j$ problem was listed as open by Khatami et al \cite{KSC20}. This paper studies this problem.  
In the problem we have a set of $n$ independent jobs $J_1, \dots, J_n$ to be scheduled on a single machine $M$. Job $J_j$ consists of two tasks (operations),
$O_{j1}$ and $O_{j2}$, the former requires processing for $p$  time units and the latter requires $b_j$ time units on $M$. Preemptions are not  permitted, and tasks $O_{j1}$ and $O_{j2}$ need to be processed in that order. There needs to be an \emph{exact} time delay of $p$ time units between the completion of $O_{j1}$ and the start  of $O_{j2}$.  The machine $M$ can process at most one task at a time. The objective is to minimize maximum lateness $L_{\max}=\max_j \{L_j=C_j-d_j\}$, where the completion time $C_j$ of $J_j$ equals the completion time of its task $O_{j2}$, and $d_j$ is due date of job $J_j$. Let us point out that 
the makespan minimization problems $1|(p, p, b_j)|C_{\max}$ and $1|(a, p, b_j)| C_{\max}$ with common delay have been shown polynomial and  NP-hard in the strong sense respectively by Orman and Potts \cite{OP97}. Furthermore,  making the delays job-dependent renders minimization of makespan NP-hard in the strong sense even for unit-time tasks, i.e. the problem  $1|(1, l_j, 1)| C_{\max}$, see Yu et al. \cite{YHL04} .
To our knowledge no problem which is polynomial for makespan but NP-hard for $L_{\max}$ has been identified so far which is a further motivation to study the problem $1|(p, p, b_j)|\max L_j$.

The paper shows polynomial time algorithms  for the \emph{agreeable} (where $d_i\geq d_j$ implies $b_i\geq b_j$) case  in Section \ref{A}, and 
the \emph{disagreeable} (where $d_i\geq d_j$ implies $b_i\leq b_j$) case in Section  \ref{D}.  Surprisingly, despite the conditions imposed on the relationship between the processing times and due dates in either  case it turns quite challenging 
to find optimal schedules that minimize maximum lateness. The polynomial time, optimization algorithms turn out quite involved for either case.  Finally, Section \ref{G}  drops any assumptions about the relationship between the processing times and due dates instead  gives a polynomial time
 algorithm for a  given partition of jobs into a subset of those whose second tasks contributed to schedule makespan and a subset of those that do not since they can be interlaced with other jobs. However the question of obtaining an
 optimal partition in polynomial time remains open. 
 %The complexity of the general problem remains open

\section{Agreeable case} \label{A}

We consider the \emph{agreeable} case where the earliest due date (EDD) order

\begin{equation}\label{edd}
d_1\leq \dots \leq d_n,
\end{equation}
of jobs implies the shortest processing time (SPT) order of second task processing times

\begin{equation} \label{spt}
b_1 \leq \dots \leq b_n.
\end{equation}
We define two linear orders on the jobs $1, \dots, n$ as follows

\begin{equation}
i<j \text{ if and only if } d_i\leq d_j,
\end{equation}
and
%We also define order on the jobs $1, \dots, n$ as follows

\begin{equation} \label{ord}
i \prec j \text{ if and only if }   d_i-b_i \leq d_j-b_j.
\end{equation}
Both orders are relevant in the construction of optimal schedules. However the orders are generally \emph{not} consistent with each other: for example for two jobs $i$ and $j$ with $d_i=5, b_i=1$ and $d_j=7, b_j=4$ respectively we have $i<j$ but $j\prec i$.
This inconsistency adds to the minimization of maximum lateness complexity. 

For instance $I=\{1, \dots, n\}$, let $T_I$ be the set of all \emph{long} jobs $j$  in $I$ with $b_j>p$.

\subsection{Pairs without long jobs}
In this section we assume $b_i\leq p$ for each $i$. The coupled tasks of a job $i$ require an exact delay of $p$. The delay may be used by a task of another job $j$,
% as long as the task in not longer that $p$
the two jobs then become \emph{interlaced} to create a pair of jobs in a schedule.
We use the notation $\{i,j\}$ for an \emph{unordered}  pair of jobs $i$ and $j$ that make up a pair of interlaced jobs in a schedule. The \emph{ordered} pairs $(i, j)$ and $(j, i)$ indicate the order of jobs within the pair $\{i,j\}$ in the schedule. The former  has the job $i$ to start (relatively) at $0$ and complete at $2p+b_i$ and the job $j$ to start at $p$ and complete at $3p+b_j$. The latter has the job $j$ to start at $0$ and complete at $2p+b_j$ and the job $i$ to start at $p$ and complete at $3p+b_i$. We allow $i=j$ in the pair $\{i,j\}$ which results in a single job $i$ (referred to as a \emph{singleton}) to start at $0$ and complete at $2p+b_i$. We also use notation $\{i<j\}$ instead of $\{i,j\}$ to indicate that $d_i\leq d_j$, however the notation does not imply the order of the interlaced jobs  in the pair. As pointed out earlier
the order defined  by (\ref{ord})  may not be consistent with the EDD order defined by (\ref{edd}) (or the SPT order defined by (\ref{spt}) for that matter). This possible discord causes that maximum lateness of two interlaced jobs
$\{i<j\}$ is minimized  by $(i, j)$ if $i \prec j$ but by $(j, i)$ if $j\prec i$. Observe that the latter minimizes makespan of the interlaced jobs at the same time: the makespan of $(i, j)$ equals to $3p+b_j$ while the makespan of $(j, i)$ equals to $3p+b_i\leq 3p+b_j$ since $d_i\leq d_j$ imples $b_i\leq b_j$ for the agreeable case.
Thus the order $(j, i)$ minimizes makespan but not necessarily maximum lateness for the pair $\{i<j\}$.
 
 We have the following  characterization of consecutive pairs of interlaced jobs in optimal schedules.

\begin{lemma}\label{L10}
For pairs $(i, j)(k,l)$ in an optimal schedule we have $\min\{d_i,d_j\} \leq \min\{d_k,d_l\}$ and $\max\{d_i,d_j\} \leq \max\{d_k,d_l\}$. 
\end{lemma}

\begin{proof} 

 \begin{figure}
\centering         %declaration corresponding to the center environment
\includegraphics[scale=0.60]{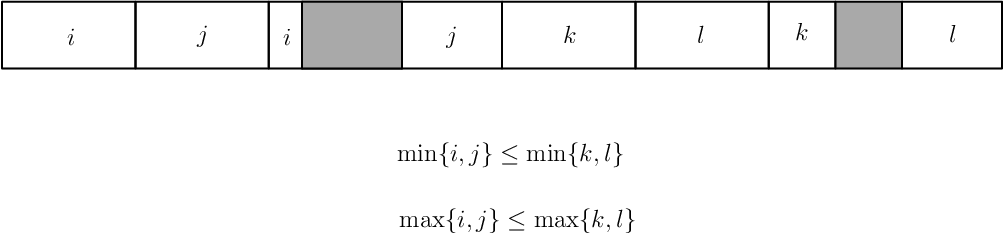}
\caption{The conditions for the arc between $(\{i, j\}, c)$ and $(\{k, l\}, c-2)$ to exist.}
\label{F0}
\end{figure}

By exchange argument. We first prove that 
\begin{equation}\label{V1}
d_j\leq d_l
\end{equation}
for  $(i, j)(k, l)$, see Figure \ref{F0}. 
Suppose 
%for contradiction 
that $d_j>d_l$.
% for some consecutive pairs in a schedule. 
% Assume that the schedule has the minimum number of such  pairs violating (\ref{V1}). 
Swap $j$ and $l$. The swap changes the lateness of jobs $j$ and $l$ from $L_j=3p+b_j-d_j$ and $L_l=6p+b_j+b_l-d_l$ respectively to
$\ell_j=6p+b_l+b_j-d_j$ and $\ell_l=3p+b_l-d_l$ respectively. Since $d_j>d_l$, we have $L_l>\max\{\ell_j,\ell_l\}$. Moreover, since $b_l\leq b_j$, the lateness of job $k$ does not increase. Finally the lateness of any other job remains unchanged.
%Thus the maximum lateness of the resulting schedule does not increase, and $(i, l)$ and $(k, j)$ meet (\ref{V1}) which reduces the number of  pairs violating (\ref{V1}) and leads to contradiction.

Second, we prove that
\begin{equation}\label{V2}
d_j\leq d_k
\end{equation}
for $(i, j)(k, l)$. Suppose
% for contradiction 
that $d_j>d_k$.
% for some pair. Assume that the schedule has the minimum number such  pairs violating (\ref{V2}). 
Swap $j$ and $k$. The swap changes the lateness of jobs $j$ and $k$ from $L_j=3p+b_j-d_j$ and $L_k=5p+b_j+b_k-d_k$ respectively to
$\ell_j=5p+b_j+b_k-d_j$ and $\ell_l=3p+b_k-d_k$ respectively. Since $d_j>d_k$, we have $L_k>\max\{\ell_j,\ell_k\}$. Moreover, since $b_k\leq b_j$, the lateness of job $l$ does not increase. Finally the lateness of any other job remains unchanged.
%Thus the maximum lateness of the resulting schedule does not increase, and $(i, k)$ and $(j, l)$ meet (\ref{V2}) which reduces the number of  pairs violating (\ref{V2}) and leads to contradiction.

Therefore we have
\begin{equation} \label{E11}
d_j\leq \min\{d_k,d_l\}\leq \max \{d_k,d_l\}
\end{equation}
for $(i, j)(k, l)$. We now prove that
\begin{equation} \label{E21}
d_i \leq \max \{d_k,d_l\}
\end{equation}
for $(i, j)(k, l)$. 
%Observe that (\ref{E2}) clearly holds if $b_j>p$ by (\ref{E1}), thus we assume $b_j\leq p$. 
The (\ref{E21}) holds 
%for any consecutive pairs $(i, j)$ and $(k,l)$ 
for $L_i \geq L_j$ since then $d_j\geq d_i$. Observe that $L_i\geq L_j$ implies $d_j-d_i+b_i-b_j\geq p$. However $b_i-b_j\leq p$ since $b_i\leq p$. Thus $d_j\geq d_i$ and hence (\ref{E21}) holds by (\ref{E11}).

Thus to complete the proof of (\ref{E21}) it remains to consider 
%consecutive pairs $(i, j)$ and $(k,l)$ with 
$L_i < L_j$. Then  $d_j-d_i+b_i-b_j < p$. Suppose
% for contradiction 
that $d_i > \max \{d_k,d_l\}$. Swap $i$ and $k$. The swap changes the lateness of jobs $i$ and $k$ from $L_i=2p+b_i-d_i$ and $L_k=5p+b_j+b_k-d_k$ respectively to
$\ell_i=5p+b_j+b_i-d_i$ and $\ell_k=2p+b_k-d_k$ respectively. 
If $L_k\geq \ell_i$, then $L_k\geq \max\{\ell_i,\ell_k\}$ hence  the maximum lateness of the resulting schedule does not increase.
% and the schedule has the pairs $(k, j)$ and $(i, l)$ which reduces the number of violating  pairs and leads to contradiction.
Observe that the lateness of any other job remains unchanged. 
%Thus it remains to consider the case 
If $L_k < {\ell}_i$, then $b_i-d_i>b_k-d_k$ which implies that $L_k<L_l$. To prove that suppose for contradiction that $L_k\geq L_l$, i.e.  $5p+b_j+b_k-d_k\geq 6p+b_j+b_l-d_l$. This implies
$b_i-d_i>b_k-d_k\geq p+b_l-d_l$. Hence $d_l-d_i+b_i-b_l\geq p$. This however leads to a contradiction since $d_l-d_i<0$ by assumption, and $b_i-b_l\leq p$ since $b_i\leq p$. Therefore we have
\begin{equation}
L_i < L_j  \text{   and  }  L_k<L_l.
\end{equation}
Since $b_i-d_i>b_k-d_k$, we have $L_j>L_i>\ell_k$. Moreover $L_l>\ell_i$ since $p+b_l-d_l>b_i-d_i$ for $d_i>d_l$ and $b_i\leq p$. 
%To see this we observe that $d_l-d_i<0$ by assumption. 
Therefore the swap keeps the maximum lateness of $(i, j)$ the same as $(k, j)$ and equal $L_j$, and the maximum lateness of $(k, l)$ the same as $(i, l)$ and equal $L_l$.
Observe that the lateness of any other job remains unchanged. Thus the maximum lateness of the resulting schedule does not increase.
% and the schedule has the pairs $(i, k)$ and $(j, l)$ which reduces the number of  pairs violating (\ref{E21}) and leads to contradiction.
\end{proof}

\begin{lemma} \label{L11X}
For pairs $\{i<j\}\{k<l\}$  we have either $i<j<k<l$ (the separated pairs)  or $i<k<j<l$ (the interlaced pairs).
\end{lemma}
\begin{proof}
Follows immediately from Lemma \ref{L10}.
\end{proof} 

We use the terms interlaced jobs as well as interlaced pairs. The terms are different but hopefully do not cause confusion. The next lemma deals with a singleton followed or preceded by a pair of interlaced jobs.

\begin{lemma} \label{L12X}
For
\begin{equation} \label{AA1}
\{i\}\{k<l\}
\end{equation}
we have $d_i\leq \min\{d_k,d_l\}$, and for
\begin{equation} \label{AA2}
\{k<l\} \{i\}
\end{equation}
we have $ \max\{d_k,d_l\} \leq d_i$.
\end{lemma}
\begin{proof} For contradiction in the proof of (\ref{AA1}) suppose $i>k$. Replace $i(k, l)$ by $k(i, l)$ which is shorter since $b_k\leq b_i$ and it does not increase maximum lateness since $d_i\geq d_k$. 
Replace $i(l, k)$ by $k(l, i)$ which has the same makespan and it does not increase maximum lateness since $d_i\geq d_k$. 

For contradiction in the proof of (\ref{AA2}) suppose $l>i$. Replace $(k, l)i$ by $(k, i)l$ which has the same makespan and it does not increase maximum lateness since $d_l\geq d_i$. Replace $(l, k)i$ by $(l, i)k$ for $k>i$. The latter
has the same makespan and it does not increase maximum lateness since $d_i\leq d_k$. It remains to consider $(l, k)i$ for $k<i$. Observe that $b_l\leq p$. Replace $(l, k)i$ by $k(l, i)$ of the same makespan. For $(l, k)i$ we have 
\begin{align*}
\ell_l=2p+b_l-d_l\\
\ell_k=3p+b_k-d_k\\
\ell_i=5p+b_k+b_i-d_i,
\end{align*}
for $k(l, i)$ we have
\begin{align*}
L_k=2p+b_k-d_k\\
L_l=4p+b_k+b_l-d_l\\
L_i=5p+b_k+b_i-d_i
\end{align*}
which does not increase maximum lateness since $\ell_i\geq L_l$ (we have $p+b_i\geq b_l +d_i-d_l$, $b_l\leq p$ and $d_i-d_l\leq 0$)

\end{proof}

\subsection{Sequences without long jobs}

By Lemmas \ref{L10}, \ref{L11X}, and \ref{L12X} we get the following characterization of optimal sequences.
A sequence $\{i_1,j_1\}, \dots, \{i_{m}, j_{m}\}$ such that $1=i_1< \dots <i_{m}$, $j_1< \dots <j_{m}=n$, $\sum_{k=1}^m|\{i_k,j_k\}|=n$, and either $i_k<j_k$ or $i_k=j_k$  for $k=1, \dots, m$,  $i_k\neq j_{k-1}$ for $k=2, \dots, m$ is called a \emph{semi-feasible }sequence.
%and $\sum _{k=1}^{\frac{n}{2}} (i_k + j_k)=\frac{n(n+1)}{2}$.  
A semi-feasible sequence is  \emph{feasible} if

\begin{equation*}
\{i_1, \dots,i_{m}\} \cup \{j_1, \dots, j_{m}\}=\{1, \dots, n\}.
\end{equation*}
%If $k=0$ the sequence is referred to as feasible. Each job in $\{i_1, \dots,i_{\frac{n}{2}}\} \cap \{j_1, \dots, j_{\frac{n}{2}}\}$ is 
%called a \emph{double}, and
In a semi-feasible sequence which is \emph{not} feasible each job in $\{1, \dots, n\} \setminus  (\{i_1, \dots, i_{m}\}\cup \{j_1, \dots, j_{m}\})$ is called a \emph{miss}, and each job $x\in \{i_k,j_k\}\cap \{i_l,j_l\}$ for some $k<l$ 
is called a \emph{double}. The number of misses equals the number of doubles in a semi-feasible sequence.

Let $s$ be a feasible sequence, a pair $s=s_1\{i<j\}\{k<l\}s_2$ such that $i<j<k<l$ is 
a \emph{separator}. Then the sequence $s_1\{i<j\}$ includes all jobs in $[1..j]$, and the sequence $\{k<l\}s_2$ includes all jobs in $[k..n]$. Thus in a feasible sequence $k=j+1$. The singleton $i$ in the sequence $s=s_1\{i,i\}=\{i\}s_2$ is also a \emph{separator}.  The sequence $s_1$ includes all jobs in $[1..i-1]$
and the sequence $s_2$ includes all jobs in $[i+1..n]$, if $i<n$. The separators decomposed  a feasible sequence into disjoint intervals $[\alpha..\beta]$, where the subsequence $s_{\alpha, \beta}$ that includes all jobs in $[\alpha..\beta]$ is made of \emph{interlaced} pairs (if $\beta-\alpha \geq 4$) or a single pair $\{\alpha, \alpha+1\}$ (if $\beta-\alpha=1$) or a singleton $\{\alpha\}$ (if $\beta=\alpha$). The subsequence $s_{\alpha, \beta}$ is called a \emph{block}.

\subsection{Order of jobs in the interlaced pairs}

We now focus on the sequence of jobs in the first pair of interlaced pairs. 

\begin{lemma} \label{L2}
For the interlaced pairs $\{i<j\}\{k<l\}$  we have $(j, i)$.
\end{lemma}

\begin{proof}
The proof is by the exchange argument. Suppose for contradiction that $(i,j)$. We have two cases. 
Case $(i, j)$ and $(k, l)$. Swap $j$ and $k$ to get $(i, k)$ and $(j, l)$. We have
\begin{align*}
L_j& = 3p + b_j -d_j\\
L_k& = 5p + b_j + b_k -d_k\\
L_l& = 6p + b_j+ b_l - d_l
\end{align*}
before the swap, and 
\begin{align*}
\ell_j& = 5p + b_k+b_j -d_j\\
\ell_k& = 3p + b_k -d_k\\
\ell_l& = 6p + b_k + b_l - d_l
\end{align*}
after the swap. We have $L_k\geq \ell_k$, and $L_k\geq \ell_j$ since $d_k\leq d_j$, and $L_l\geq \ell_l$ since $b_k\leq b_j$ for the agreeable case. Moreover, since $b_k\leq b_j$ the schedule of the two pairs $(i, k)$ and $(j, l)$ is not longer ($b_j+b_l$ before, and $b_k+b_l$ after) than the initial schedule of $(i, j)$ and $(k, l)$ before that swap.

Case $(i, j)$ and $(l, k)$, swap $j$ and $k$ to get $(i, k)$ and $(l, j)$. We have
\begin{align*}
L_j& = 3p + b_j -d_j\\
L_k& = 6p + b_j + b_k -d_k\\
L_l& = 5p + b_j + b_l - d_l
\end{align*}
before the swap, and 
\begin{align*}
\ell_j& = 6p + b_k+b_j -d_j\\
\ell_k& = 3p + b_k -d_k\\
\ell_l& = 5p + b_k + b_l - d_l
\end{align*}
after the swap. We have $L_k\geq \ell_k$, and $L_k\geq \ell_j$ since $d_k\leq d_j$, and $L_l\geq \ell_l$ since $b_k\leq b_j$ for the agreeable case. 
The swap does not change the makespan of the pairs.
%Moreover, since $b_k\leq b_j$ the schedule of the two pairs $(i, k)$ and $(l, j)$ is not longer ($b_j+b_k$ before, and $b_k+b_j$ after) than the initial schedule of $(i, j)$ and $(l, k)$ before that swap.

\end{proof}

\subsection{Sequences with long jobs}

This section considers instances with nonempty $T_{\mathcal{I}}$. We first identify the schedules which can be eliminated from consideration without loosing generality. Let $S$ be a schedule.

\begin{lemma} \label{L13A}
Without loss of generality we can assume that $i(j,k)$, where $i\in T_{\mathcal{I}}$ does not occur in $S$.
\end{lemma}
\begin{proof} Suppose that $i(j,k)$, where $i\in T_{\mathcal{I}}$ occurs in $S$, we have
%Before $i(j,k)$
\begin{align*}
\ell_i=2p+b_i-d_i\\
\ell_j=4p+b_i+b_j-d_j\\
\ell_k=5p+b_i+b_k-d_k.
\end{align*}
In $S$ change $i(j,k)$ to $(j,i)k$. For the latter we have
\begin{align*}
L_j=2p+b_j-d_j\\
L_i=3p+b_i-d_i\\
L_k=5p+b_i+b_k-d_k
\end{align*}
Thus $\ell_k= L_k$. Moreover $\ell_j>L_j$, and since $i>j$, we have $\ell_j>L_i$. Therefore $\max\{L_i,L_j,L_k\}\leq \max\{\ell_i,\ell_j,\ell_k\}$. Moreover, the makspans of $i(j,k)$ and $(j,i)k$ are the same.

%Since $j>l$, we have $l_l\geq L_j$. Also $l_l>L_l$. Since $j>l$ and $j\in H_{\mathcal{I}}$, we have $l\in H_{\mathcal{I}}$. By Lemma \ref{L12} $k>l$, thus
%$l_l\geq L_k$. Thus $\max\{l_j,l_k,l_l\}\geq l_l\geq \max\{L_j,L_k,L_l\}$.
\end{proof}

\begin{lemma}\label{L15A}
Without loss of generality we can assume that $ij$, where $i\in T_{\mathcal{I}}$ and $j\in \mathcal{I}\setminus T_{\mathcal{I}}$ does not occur in $S$.
\end{lemma}

\begin{proof}
If $j\in \mathcal{I}\setminus T_{\mathcal{I}}$, then the sequence $ij$ of two singletons  can be replaced by a shorter pair $(j,i)$ of interlaced jobs $i$ and $j$ without increasing maximum lateness since $d_j\leq d_i$ in the agreeable case.
\end{proof}

By lemmas \ref{L13A} and \ref{L15A}  it is sufficient to consider schedule $S$ with singleton jobs from $T_{\mathcal{I}}$, if any,  all pushed to the end of $S$. Moreover, it can be readily shown that those singleton long jobs are  scheduled in the EDD order in $S$. Thus, we just shown that the following lemma holds.

\begin{lemma}
Let $S$ has singletons  from $T_{\mathcal{I}}$. Then, a singleton $i\in T_{\mathcal{I}}$ is either last in $S$ or it is followed by another singleton  $j\in T_{\mathcal{I}}$, if any, such that $d_i\leq d_j$.
\end{lemma}
%\begin{proof}
%\end{proof}

We now turn our attention to the long jobs in $T_{\mathcal{I}}$ which may be interlaced with other jobs in $S$.

\begin{lemma} \label{L13C}
Without loss of generality we can assume that $(i,j)(k,l)$, where $j\in T_{\mathcal{I}}$ and $l\in \mathcal{I}\setminus T_{\mathcal{I}}$ does not occur in $S$.
\end{lemma}
\begin{proof} Suppose that $(i,j)(k,l)$, where $j\in T_{\mathcal{I}}$ and $l\in \mathcal{I}\setminus T_{\mathcal{I}}$ occurs in $S$. For $(i,j)(k,l)$ we have
\begin{align*}
\ell_i=2p+b_i-d_i\\
\ell_j=3p+b_j-d_j\\
\ell_k=5p+b_j+b_k-d_k\\
\ell_l=6p+b_j+b_l-d_l.
\end{align*}
For $(i,l)(k,j)$ we have
\begin{align*}
L_i=2p+b_i-d_i\\
L_l=3p+b_l-d_l\\
L_k=5p+b_l+b_k-d_k\\
L_j=6p+b_l+b_j-d_j
\end{align*}
Since $j>l$, we have $\ell_l\geq L_j$, and $\ell_k\geq L_k$.  Also $\ell_l\geq L_l$, and $\ell_i\geq L_i$. Thus $\max\{\ell_i,\ell_j,\ell_k,\ell_l\}\geq \max\{\ell_k,\ell_l\}\geq \max\{L_i,L_j,L_k,L_l\}$.

\end{proof}

\begin{lemma}\label{L15C}
Without loss of generality we can assume that $(i,j)k$, where $j\in T_{\mathcal{I}}$ and $k\in \mathcal{I}\setminus T_{\mathcal{I}}$, does not occur in $S$.
\end{lemma}
\begin{proof} Suppose that $(i,j)k$, where $j\in T_{\mathcal{I}}$ and $k\in \mathcal{I}\setminus T_{\mathcal{I}}$ occurs in $S$. For $(i,j)k$ we have
\begin{align*}
\ell_i=2p+b_i-d_i\\
\ell_j=3p+b_j-d_j\\
\ell_k=5p+b_j+b_k-d_k.
\end{align*}
and for $(i,k)j$ we have
\begin{align*}
L_i=2p+b_i-d_i\\
L_k=3p+b_k-d_k\\
L_j=5p+b_k+b_j-d_j.
\end{align*}
Since $j>k$, we have $\ell_k\geq L_j$. Also $\ell_k\geq L_k$, and $\ell_i\geq L_i$. Thus $\max\{\ell_i,\ell_j,\ell_k\}\geq \max\{L_i,L_j,L_k\}$.

\end{proof}

By Lemmas \ref{L13C} and \ref{L15C}  it is sufficient to consider schedule $S$ with pairs of interlaced jobs where one of the interlaced jobs is from $T_{\mathcal{I}}$, if any,  all pushed to the end of $S$ yet, by Lemma \ref{L13A},  before singleton jobs from $T_{\mathcal{I}}$, if any.

We can summarize all those observations in the following lemma.

\begin{lemma} \label{LT}
For an instance $\mathcal{I}$ with nonempty $T_{\mathcal{I}}=\{i, \dots, n\}$ there is an optimal schedule $S$ that ends with
\begin{equation}
(y_a,i)\dots(y_1, i+a-1) i+a \dots n
\end{equation}
where $y_a\prec \dots \prec y_1$ are in $\mathcal{I}\setminus T_{\mathcal{I}}$, and $0\leq a\leq n-i+1$. That is
\begin{equation*}
S=s(y_a,i)\dots(y_1, i+a-1) i+a \dots n,
\end{equation*}
where $s$ is the schedule for an instance $\mathcal{I}\setminus \{y_1, \dots, y_a, i, \dots, n\}$ without long jobs.
\end{lemma}

%In order to find an optimal schedule it suffices to consider schedules  for each $0\leq a\leq n-i+1$. 
%The jobs $y_1, \dots, y_a$ can not be decided in advance since the orders $\prec$ and $<$ may not be consistent. 
Observe  that the border pairs $(k<l)(y_a<i)$ 
between $s$ and $(y_a,i)\dots(y_1, i+a-1)$, where $y_a,k,l\in \mathcal{I}\setminus T_{\mathcal{I}}$, $i\in T_{\mathcal{I}}$ may be interlaced since $l\prec y_a$ though $y_a<l$.
Thus $(y_a,i)\dots(y_1, i+a-1) $ may not be a block since $(k<l)(y_a<i)$ may not be a separator and consequently the jobs $y_1, \dots, y_a$ can not be decided in advance and need to be determined by the algorithm given in subsequent sections.

Finally, observe that the schedule $s(y_a,i)\dots(y_1, i+a-1) $ does not necessarily satisfy Lemma \ref{L10} which is an anomaly resulting from the presence of long jobs.
To see this let $\{i<j\}$ immediately precedes  $\{k<l\}$ such that $j, l\in T_{\mathcal{I}}$ in a schedule. Clearly $(i, j)$ and $(k, l)$. Also we have $j<l$ for the agreeable case, and $i\prec k$. However 
we may have $i>k$ at the same time. Thus $\min\{d_i,d_j\}\geq \min\{d_k,d_l\}$ and $\max\{d_i,d_j\}\leq\max\{d_k,d_l\}$ which is in contrast to Lemma \ref{L10} for the schedules without tails.

\subsection{From semi-feasible to feasible sequences}

%Thus we have $\min\{d_i,d_j\}\leq \min\{d_k,d_l\}$ and $\max\{d_i,d_j\}\leq\max\{d_k,d_l\}$ for the border. We also have $j\prec k$. However we may have $j>k$, which meas that $\{i<j\}$ and $\{k,l\}$ are interlaced. 

We now show that a semi-feasible sequence can always be turned into a feasible one without loosing optimality. We have the following key lemma.
%With these observations in mind we get the following key result.

\begin{lemma} \label{LM}
Let  $\alpha$ be a semi-feasible sequence with maximum lateness not exceeding $\lambda$.  For $\alpha$ there is a feasible sequence $\beta$   with maximum lateness $\lambda'\leq \lambda$.
\end{lemma}
\begin{proof}  %The lemma holds for $k=0$. Thus assume $k>0$. Thus $\alpha$ has $k$ doubles and the same number of misses. 
 Consider $\alpha$ with minimum number of doubles. 
% If there are no doubles, then $\alpha$ is feasible and the lemma holds. 
For contradiction suppose that  $i_a$ is a double in $\alpha$.
%there is at least one double in $\alpha$. 
Without loss of generality it suffices to consider $\alpha$ being a block with at least two pair ($\beta-\alpha\geq 4$). Suppose first that $j_a\in \mathcal{I}\setminus T_{\mathcal{I}}$.
Let $i_a = j_b$ for some $b$.  We then have $a>b$ since $a<b$ implies $i_a<j_a<j_b=i_a$ which gives contradiction. Thus for $\alpha$ we have
\begin{equation*}
i_1<\dots < i_b<\dots < i_a, \dots ,i_{m} 
\end{equation*}
\begin{equation*}
j_1<\dots < j_b= i_a<\dots < j_a,\dots ,j_{m}.
\end{equation*}

The pairs $\{i_c,j_c\}$ and $\{i_{c+1}, j_{c+1}\}$ for $c=1, \dots, a-1$ are interlaced. Observe that by definition of semi-feasible sequence we have $a-1>1$. By Lemma \ref{L2} we have
\begin{equation*}
(j_1, i_1) \dots (j_b, i_b) \dots (j_{a-1}, i_{a-1})
%, \dots, j_a\rightarrow i_a.
\end{equation*}
in $\alpha$.

 \begin{figure}
\centering         %declaration corresponding to the center environment
\includegraphics[scale=0.60]{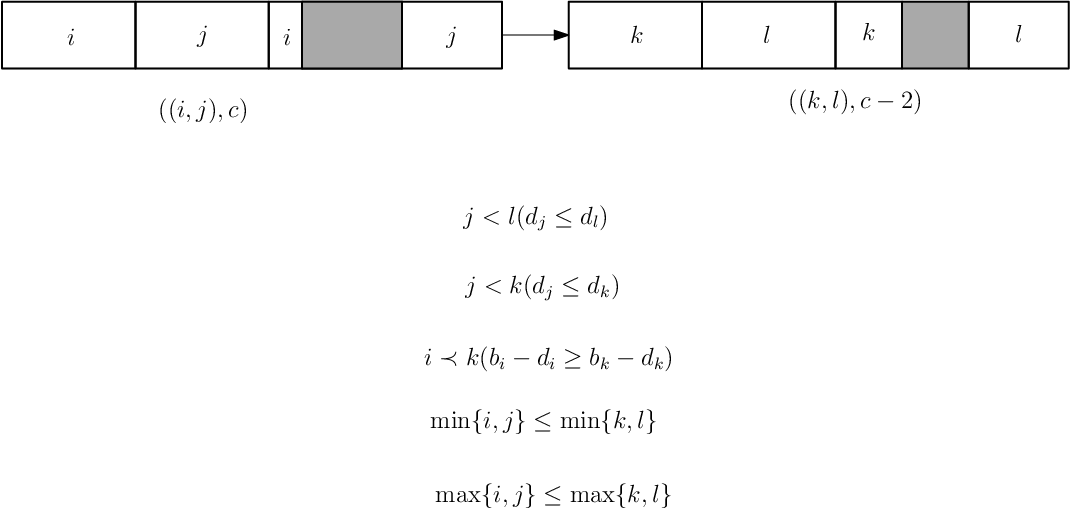}
\caption{The conditions for the arc between $((i, j), c)$ and $((k, l), c-2)$ to exist.}
\label{F1}
\end{figure}
Let $x$ be the largest index such that
\begin{equation}
j\in \{i_x, \dots, n\} \setminus ( \{i_x, \dots, i_{m}\}\cup \{j_x, \dots, j_{m}\}),
\end{equation}
such an index exists since there is a miss in $\alpha$. Assume $x\leq b$, replace
%\begin{equation*}
%i_1<\dots < i_x< \dots < i_b<\dots < i_a<\dots <i_{\frac{n}{2}} 
%\end{equation*}
\begin{equation*}
(j_1, i_1) \dots (j_x, i_x) \dots  (j_b= i_a, i_b)\dots  (j_{a-1}, i_{a-1})
\end{equation*}
in $\alpha$ by
\begin{equation*}
(j_1, i_1) \dots (j, i_x) \dots  (j_{b-1}, i_b)\dots  (j_{a-1}, i_{a-1})
\end{equation*}
to obtain sequence $\beta$.

We have $3p+b_{i_b}-d_{i_b}\geq 2p+b_{j_{b-1}}-d_{j_{b-1}}$ since $d_{i_b}\leq d_{j_{b-1}}$ (interlacing) and $b_{j_{b-1}}-b_{i_b}\leq p$ (assumption). Observe that this applies to all pairs between $x$ and $b$, if any,  since $i_{\ell}<j_{\ell-1}$ 
(otherwise, if $i_{\ell}>j_{\ell-1}$ for some $\ell$, then the pairs $\{i_{\ell-1}<j_{\ell-1}\}$ and  $\{i_{\ell}<j_{\ell}\}$ would not be interlaced which gives contradiction) for $\ell=x+1, \dots,b$. Thus the replacement does not increase the maximum lateness and preserves the makespan.
% since
%\begin{equation*}
%(j, i_x), (j_x, i_{x+1}), \dots, (j_{b-1}, i_b)
%, \dots, j_a\rightarrow i_a.
%\end{equation*}
Therefore we obtain a sequence $\beta$ with fewer doubles than in $\alpha$ which gives contradiction.
Now assume $b<x$, replace
%\begin{equation*}
%i_1<\dots < i_x< \dots < i_b<\dots < i_a<\dots <i_{\frac{n}{2}} 
%\end{equation*}
\begin{equation*}
(j_1, i_1) \dots   (j_b= i_a, i_b)\dots  (j_{a-1}, i_{a-1})
\end{equation*}
in $\alpha$ by
\begin{equation*}
(j_1, i_1) \dots  (j, i_b)\dots  (j_{a-1}, i_{a-1})
\end{equation*}
to obtain sequence $\beta$. We have $3p+b_{i_b}-d_{i_b}\geq 2p+b_{j}-d_{j}$ since $d_{i_b}\leq d_{j}$ and $b_{j_{b-1}}-b_{i_b}\leq p$ (assumption). Thus the replacement does not increase the maximum lateness and preserves the makespan.
% since
%\begin{equation*}
% (j, i_b)
%, \dots, j_a\rightarrow i_a.
%\end{equation*}
Therefore we obtain a sequence $\beta$ with fewer doubles than in $\alpha$ which gives contradiction.

Suppose now $j_a\in T_{\mathcal{I}}$. Then $(i_{a-\ell}, j_{a-\ell}), \dots, (i_a, j_a)$, where $j_{a-\ell}$ is the earliest in $T_{\mathcal{I}}$, and by Lemma \ref{L2} $(j_b, i_b) \dots (j_{a-\ell-1}, i_{a-\ell-1})$ with all jobs
in $\mathcal{I}\setminus T_{\mathcal{I}}$.
However $j_b=i_a\prec \dots \prec  j_{a-\ell-1}\prec i_{a-\ell}\prec \dots \prec i_a$ which gives contradiction.

\end{proof}

\subsection{The graphs for the agreeable case}

We begin with the graph $G=(\mathcal{N},A)$  of \emph{unordered} pairs. The set of nodes $\mathcal{N}$
%$\mathcal{N}=\{(\{i< j\}, c): 1\leq i< j \leq n, i\in \mathcal{I}\setminus T_{\mathcal{I}}, 2\leq c \leq n\}\cup \{(\{i,i\}, c): 1\leq i \leq n, 1\leq c \leq n\}$.
%  1\leq c \leq \frac{n}{2}, \frac{n}{2}-i+1\leq c \leq \frac{n-i+1}{2}$.
% \text{if  } b_j>p, \text{ then  } b_i\leq p\}$. 
\begin{equation*}
\mathcal{N}=\{(\{i< j\}, c): 1\leq i< j \leq n, i\in \mathcal{I}\setminus T_{\mathcal{I}}, 2\leq c \leq n\}\cup \{(\{i,i\}, c): 1\leq i \leq n, 1\leq c \leq n\}.
\end{equation*}
The set of arcs $A$

For $j\in \mathcal{I}\setminus T_{\mathcal{I}}$, the arcs from $(\{i<j\}, c\geq 2)\in \mathcal{N}$ exist to each  $(\{k< l\}, c-2)\in \mathcal{N}$
such that $k>i$, $l>j$ and $j\neq k$, and to   $(\{j+1, j+1\}, c-2)\in \mathcal{N}$.

For $j\in T_{\mathcal{I}}$, the arcs from $(\{i<j\}, c\geq 2)\in \mathcal{N}$ exist to each  $(\{k< j+1\}, c-2)\in \mathcal{N}$
such that $i\prec k$, and to   $(\{j+1, j+1\}, c-2)\in \mathcal{N}$.

For $i\in \mathcal{I}\setminus T_{\mathcal{I}}$, the arcs from $(\{i,i\}, c\geq 1)\in \mathcal{N}$ exist to each  $(\{i+1< l\}, c-1)\in \mathcal{N}$
%, and to  $(\{i+1,i+1\}, c-1)\in \mathcal{N}$
provided that $i+1\in \mathcal{I}\setminus T_{\mathcal{I}}$,
%. If $i+1\in  T_{\mathcal{I}}$, the arc from $(\{i,i\}, c\geq 1)\in \mathcal{N}$ exists 
and to  $(\{i+1,i+1\}, c-1)\in \mathcal{N}$.

For $i\in T_{\mathcal{I}}$, the arc from $(\{i,i\}, c\geq 1)\in \mathcal{N}$ exists to  $(\{i+1,i+1\}, c-1)\in \mathcal{N}$ provided that $i<n$.

No other arcs exist in $A$. 

%Let $G=(\mathcal{N},A)$ be the graph of unordered pairs. 
In the graph $G=(\mathcal{N},A)$ we focus on the paths from the \emph{initial} nodes $(\{1\leq j\}, n)$, $j=1, \dots, n$, of in-degree zero to
the \emph{final} nodes $(\{i\leq n\},|\{i,n\}|)$, $i=1, \dots, n$, of out-degree zero. Any such path is of the following form 
\begin{equation*}
(\{i_1=1,j_1\}, c_1=n)\rightarrow (\{i_2,j_2\}, c_2) \rightarrow \dots \rightarrow (\{i_m,j_m=n\},c_m),
\end{equation*}
where
\begin{equation}
i_1< \dots <i_{m'}, i_{m'+1}\prec \dots \prec i_{m}
\end{equation}
\begin{equation}
 j_1< \dots <j_{m'}< j_{m'+1}<\dots< j_{m}
 \end{equation}
 for some $m'\leq m$,
 \begin{equation}
 \{ j_{m'+1}, \dots, j_{m}\}=T_{\mathcal{I}}
 \end{equation}
 and
 \begin{equation}
i_k\leq j_k,
 \end{equation}
for $k=1, \dots, m$, and
 \begin{equation}
i_k\leq j_{k-1},
 \end{equation}
for $k=2, \dots, m$, and
\begin{equation}
c_k=n - \sum_{\ell=1}^{k-1} |\{ i_{\ell}, j_{\ell}\}|
\end{equation}
for $k=1, \dots, m-1$, and
\begin{equation}
c_m=n - \sum_{\ell=1}^{m-1} |\{ i_{\ell}, j_{\ell}\}|=| \{i_m,j_m\}|
\end{equation}
These seven conditions follow immediately from the definition of  $G=(\mathcal{N},A)$. A \emph{feasible} path needs to also satisfy the following condition
\begin{equation}
\{i_1, \dots,
i_{m}\}\cup \{j_1, \dots,
j_{m}\}=\{1, \dots, n\}
\end{equation}
By Lemma \ref{LM}, this last condition holds 
%sets $\{i_1, \dots,
%i_{\frac{n}{2}}\}$ and $\{j_1, \dots,
%j_{\frac{n}{2}}\}$  are disjoint 
for optimal paths in the graph of \emph{ordered} pairs which we now define.  In the graph $G$, replace each node $(\{i<j\},c)$ with two nodes $((i, j), c)$ and $((j, i), c)$, see Figure \ref{F2},
%replace each node $(\{i<j\},c)$, where $j\in T_{\mathcal{I}}$, with one node $(i \rightarrow j, c)$.
  and  each arc $(\{i<j\},c)\rightarrow (\{k<l\},c-2)$ by four arcs as in Figure \ref{F3},  each arc $(\{i,i\},c)\rightarrow (\{i+1<l\},c-1)$ by two arcs as in Figure \ref{F3B}, each arc $(\{i<j\},c)\rightarrow (\{j+1,j+1\},c-2)$ by two arcs as in Figure \ref{F3A}, to obtain the graph $G'$ of ordered pairs.

 \begin{figure}
\centering         %declaration corresponding to the center environment
\includegraphics[scale=0.60]{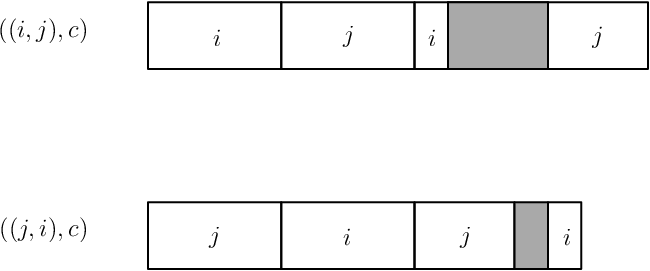}
\caption{The two nodes $((i, j), c)$ and $((j, i), c)$ (and their respective schedules) to replace $(\{i<j\}, c)$, $i, j\in\mathcal{I}\setminus T_{\mathcal{I}}$.}
\label{F2}
\end{figure}

 \begin{figure}
\centering         %declaration corresponding to the center environment
\includegraphics[scale=0.60]{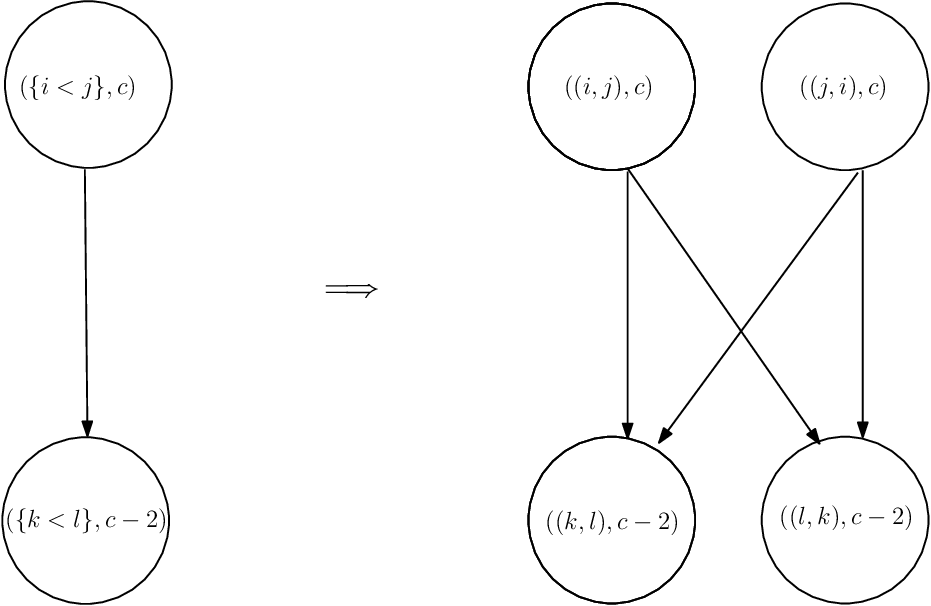}
\caption{The construction of graph $G'$.}
\label{F3}
\end{figure}

 \begin{figure}
\centering         %declaration corresponding to the center environment
\includegraphics[scale=0.60]{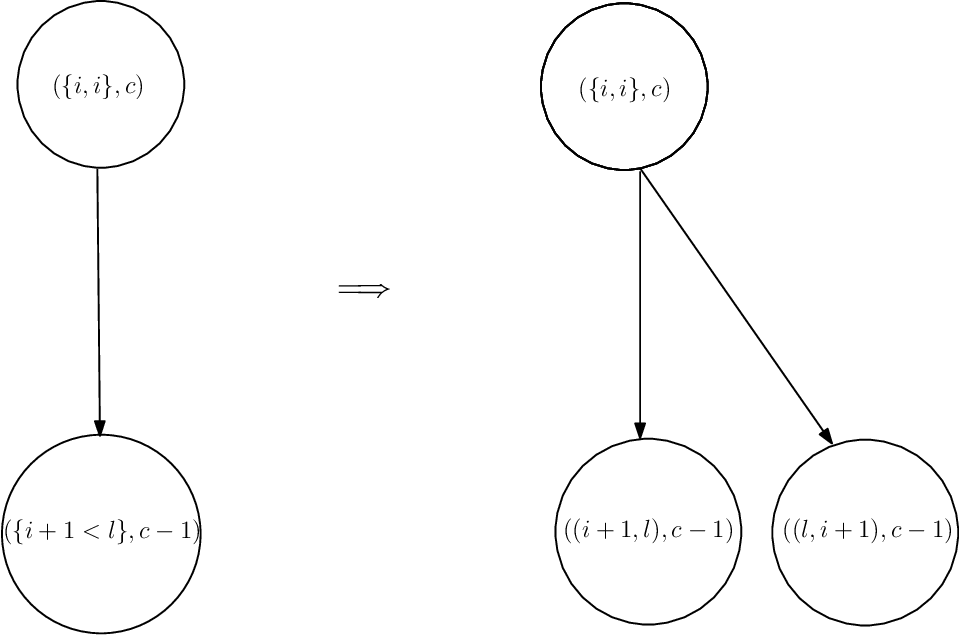}
\caption{The construction of graph $G'$.}
\label{F3B}
\end{figure}

 \begin{figure}
\centering         %declaration corresponding to the center environment
\includegraphics[scale=0.60]{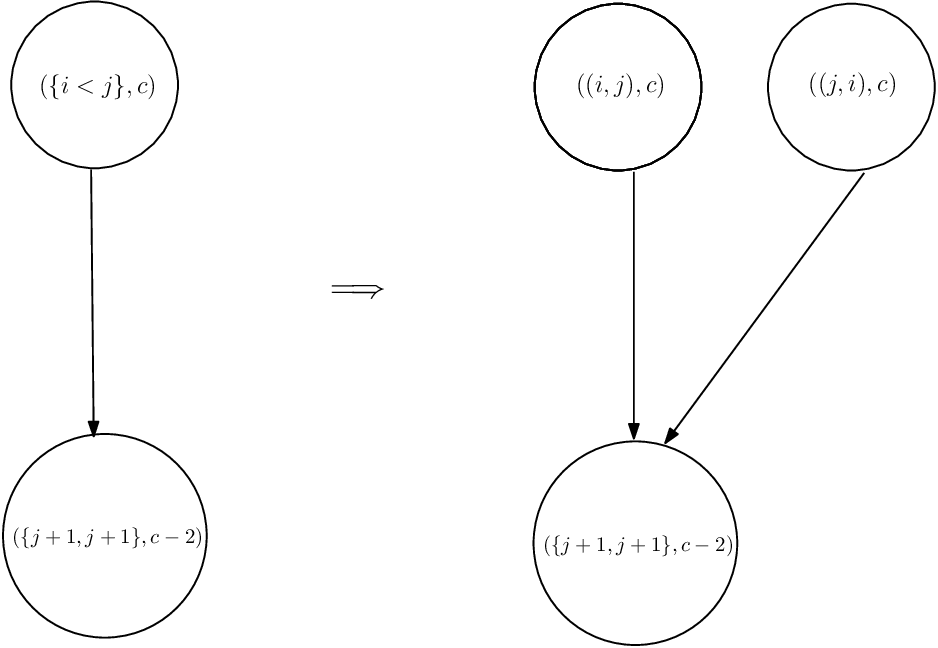}
\caption{The construction of graph $G'$.}
\label{F3A}
\end{figure}

Observe that $G'$ may include some \emph{redundant} arcs (those arcs that do not satisfy one or more conditions in Figure \ref{F1}), those arcs can be removed since they are not required by optimal paths in $G'$. As well
$G'$ may include some \emph{infeasible} nodes, those are the nodes $(j, i)$, where $j\in T_{\mathcal{I}}$ which need to be removed along with all in-coming and out-going arcs from $G'$.
Let $\mathcal{L}$ be the resulting graph, in the graph define the \emph{relative maximum lateness}
\begin{equation}
L_{((i, j),c)}=\max\{2p+b_i-d_i, 3p+b_j-d_j\}=2p+\max\{b_i- d_i, p+b_j-d_j\}
\end{equation}
of the node $((i, j), c)$, and
\begin{equation}
L_{\{i, i\},c)}=2p+b_i- d_i
\end{equation}
of the singleton node $(\{i,i\}, c)$.
For each arc going-out of the node $((i, j), c)$ define the arc \emph{length}
\begin{equation}
P_{((i, j), c)}=3p+b_j,
\end{equation}
for each arc going-out of the node $(\{i,i\}, c)$ define the arc \emph{length}
\begin{equation}
P_{(\{i,i\}, c)}=2p+b_i.
\end{equation}

For convenience add the starting node $s$ to $\mathcal{L}$ and connect it to all \emph{initial} nodes $(\{1,1\}, n)$, $((j,1), n)$, $((1, j), n)$, $j=1, \dots, n$ by the arcs of length $0$.

\subsection{The algorithm for the agreeable case}

The algorithm uses $\mathcal{L}$ doing binary search for minimum maximum lateness $\lambda$. Given $\lambda$. For a node $((i, j), c)$ in $\mathcal{L}$ let $\ell_{((i, j), c)}$ be the length of the shortest path from $s$ to $((i, j), c)$ in $\mathcal{L}$. If
\begin{equation*}
\ell_{((i, j), c)}+L_{((i, j),c)}> \lambda,
\end{equation*}
then remove $((i, j), c)$ along with all in-coming and out-going arcs from the node from $\mathcal{L}$. Similarly, for a node $(\{i,i\}, c)$ in $\mathcal{L}$ let $\ell_{(\{i,i\}, c)}$ be the length of the shortest path from $s$ to $(\{i,i\}, c)$ in $\mathcal{L}$. If
\begin{equation*}
\ell_{(\{i,i\}, c)}+L_{(\{i,i\},c)}> \lambda,
\end{equation*}
then remove $(\{i,i\}, c)$ along with all in-coming and out-going arcs from the node from $\mathcal{L}$. Let the resulting graph be $\mathcal{L}'$. If \emph{no} path from $s$ to some  \emph{final} nodes $(\{n, n\}, 1)$, $(i, n), 2)$, $((n,i), 2)$ for $i=1, \dots, n$
in $\mathcal{L}'$ exists, then answer that NO and stop.  Otherwise,  a shortest path from $s$ to some final node in $\mathcal{L}'$ gives  by Lemma \ref{LM} the schedule with maximum lateness not exceeding $\lambda$ and the algorithm stops.

\begin{lemma}
The algorithm finds a feasible schedule with maximum lateness not exceeding $\lambda$ if one exists.
\end{lemma} 
\begin{proof}  Let $S$ be a feasible schedule with maximum lateness not exceeding $\lambda$. Suppose for contradiction that the algorithm stops with $\mathcal{L}'$ not having any path from $s$ to any final node. The $S$ defines a sequence
\begin{equation*}
\pi_1 \dots \pi_m,
\end{equation*}
where each $\pi_k$ is either an interlaced pair of jobs $(i,j)$ or a singleton $i$. 
Let $C_{\pi_k}$ be the start of $\pi_k$ in $S$, and let $c_k$ be the number of jobs  that start at or after $C_{\pi_k}$ in $S$. Thus the nodes $(\pi_k, c_k)$ are in the graph $\mathcal{L}$. Moreover by
Lemmas \ref{L10} and \ref{LT} $(\pi_1,c_1)$ is one of the initial nodes and $(\pi_m,c_m)$ is one of the final nodes in $\mathcal{L}$.
Finally 
\begin{equation*} \label{EQ}
C_{(\pi_k, c_k)}+L_{(\pi_k, c_k)}\leq \lambda,
\end{equation*}
for $k=1, \dots, m$ in $S$. Since $C_{(\pi_k, c_k)}\geq \ell_{(\pi_k, c_k)}$ we have
\begin{equation*} 
\ell_{(\pi_k, c_k)}+L_{(\pi_k, c_k)}\leq \lambda,
\end{equation*}
for $k=1, \dots, m$. Therefore the path
\begin{equation*}
s(\pi_1, c_1) \dots (\pi_m, c_m),
\end{equation*}
from $s$ to a final node $(\pi_m, c_m)$ exists in $\mathcal{L}'$ which gives contradiction.

\end{proof}

\section{Disagreeable case} \label{D}
We now turn our attention to the \emph{disagreeable} case where the EDD order of jobs implies the longest processing time (LPT) order of processing times
%Order the jobs in non-increasing order of due dates, i.e. earliest due date order
%\begin{equation}
%d_1\leq \dots \leq d_n.
%\end{equation}
%We define a linear order on the jobs $1, \dots, n$ as follows

%\begin{equation}
%i<j \text{ if and only if } d_i\leq d_j.
%\end{equation}
%We consider an 

\begin{equation}
b_1 \geq \dots \geq b_n.
\end{equation}
It can be easily checked that the orders (\ref{edd}) (i.e. EDD) and (\ref{ord}) (i.e. $\prec$) are consistent for the disagreeable case. 
For an instance $\mathcal{I}=\{1, \dots, n\}$ let $H_{\mathcal{I}}$ be the set of all long jobs  $i\in \mathcal{I}$ with $b_i>p$.

%\subsection{Characteristics of schedules}

\subsection{Sequence for instance without long jobs}
We assume instances with no long jobs in this section. The order of two interlaced jobs in a pair is determined by the following lemma.

\begin{lemma} \label{L0}
For $(i, j)$ we have $i<j$.
%, and $L_i\geq L_j$.
\end{lemma}

\begin{proof}
We have $i<j$  since $b_j-b_i+d_i-d_j\leq 0$ (i.e. $i\prec j)$ in the disagreeable case. 
%Moreover $L_i=2p+b_i-d_i\geq L_j=3p+b_j-d_j$ since by the subcase assumption $p< b_i-b_{i+1}+d_{i+1}-d_i$ for  $i=1, \dots, n-1$
\end{proof}

Observe that the order $(i,j)$ minimizes the makespan of the interlaced pair $i$ and $j$, however which of the two jobs defines the maximum lateness of the pair depends on the distance $d_j-d_i$. For $d_j-d_i\leq p+b_j-b_i$ we have $L_i\leq L_j$, but for $d_j-d_i> p+b_j-b_i$ we have $L_i>L_j$.

The first jobs of the pairs follow the EDD order which is given in the following lemma.

\begin{lemma} \label{L11}
For $(i, j)(k, l)$ we have  $i<k$.
\end{lemma}
\begin{proof} Suppose $k<i$. Swap $i$ and $k$. Before the swap $L_i=2p+b_i-d_i$, and $L_k=5p+b_j+b_k-d_k$. After the swap $\ell_k=2p+b_k-d_k$, $\ell_i=5p+b_j+b_i-d_i$.
We have $L_k>\ell_k$ and $L_k\geq \ell_i$ since $k\prec i$ for the disagreeable case.
%By the assumption $k<j$, thus $l_k\geq l_j$ by the case assumption. Since $b_j\geq b_l$, we get $L_k\geq l_k$. Finally, $L_i\geq L_j$ and $L_i\geq l_l$, and $L
%_i$ does not change. 
Thus the swap does not increase the maximum lateness, and preserves the makespan. Therefore we have $i<k$.
%<l<j$ for $j>k$.
\end{proof}

Any two singletons can be interlaced.

\begin{lemma}\label{L15}
A sequence of two singletons $ij$ does not occur in $S$.
\end{lemma}

\begin{proof}
It can be replaced by shorter $(i, j)$ without increasing maximum lateness. Observe that the lemma holds also for $j\in H_{\mathcal{I}}$.
\end{proof}
Any singleton job can be pushed to the end of the sequence.

\begin{lemma} \label{L14}
Without loss of generality $i(j, k)$ does not occur in $S$.
\end{lemma}
\begin{proof} If $i<j$. For  $i(j, k)$ we have
\begin{align*}
\ell_i=2p+b_i-d_i\\
\ell_j=4p+b_i+b_j-d_j\\
\ell_k=5p+b_i+b_k-d_k,
\end{align*}
for $(i, j)k$ we have
\begin{align*}
L_i=2p+b_i-d_i\\
L_j=3p+b_j-d_j\\
L_k=5p+b_j+b_k-d_k
\end{align*}
Thus $\ell_k\geq L_k$ since $i<j$. Moreover $\ell_j>L_j$ and $\ell_i=L_i$. Therefore $\max\{L_i,L_j,L_k\}\leq \max\{\ell_i,\ell_j,\ell_k\}$.

If $j<i$. For $i(j, k)$ we have
\begin{align*}
\ell_i=2p+b_i-d_i\\
\ell_j=4p+b_i+b_j-d_j\\
\ell_k=5p+b_i+b_k-d_k,
\end{align*}
and for $(j, i)k$ we have
\begin{align*}
L_j=2p+b_j-d_j\\
L_i=3p+b_i-d_i\\
L_k=5p+b_i+b_k-d_k
\end{align*}
Thus $\ell_k= L_k$. Moreover $\ell_j>L_j$ and $\ell_j>L_i$ since $j<i$. Therefore $\max\{L_i,L_j,L_k\}\leq \max\{\ell_i,\ell_j,\ell_k\}$.

%Since $j>l$, we have $l_l\geq L_j$. Also $l_l>L_l$. Since $j>l$ and $j\in H_{\mathcal{I}}$, we have $l\in H_{\mathcal{I}}$. By Lemma \ref{L12} $k>l$, thus
%$l_l\geq L_k$. Thus $\max\{l_j,l_k,l_l\}\geq l_l\geq \max\{L_j,L_k,L_l\}$.
\end{proof}

We summarize those observations in the following lemma.

\begin{lemma} \label{LS}
There is at most one singleton in $S$, and if there is one it is the last in $S$.
\end{lemma}
\begin{proof} Follows from Lemmas \ref{L14} and \ref{L15}.
\end{proof}

\subsection{Sequences for instance with long jobs}

\begin{lemma} \label{L12}
If $(i, j)$ in $S$, where $j\in H_{\mathcal{I}}$, then $i>j$. Also $L_j>L_i$ since $p+d_i-d_j\geq b_i-b_j$ for the disagreeable case.
\end{lemma}

\begin{proof}
Otherwise $i\in H_{\mathcal{I}}$, and the jobs $i$ and $j$ cannot be interlaced in the disagreeable case.
\end{proof}

The pair of interlaced jobs with one job from $H_{\mathcal{I}}$ precedes the pair of interlaced jobs with both jobs in $\mathcal{I}\setminus H_{\mathcal{I}}$.

\begin{lemma} \label{L13}
If $(i, j)(k, l)$ in $S$, where $j\in \mathcal{I}\setminus H_{\mathcal{I}}$, then $l\in \mathcal{I}\setminus H_{\mathcal{I}}$ .
\end{lemma}
\begin{proof} Suppose for contradiction that $l\in H_{\mathcal{I}}$. For $(i, j)(k, l)$ we have
\begin{align*}
\ell_i=2p+b_i-d_i\\
\ell_j=3p+b_j-d_j\\
\ell_k=5p+b_j+b_k-d_k\\
\ell_l=6p+b_j+b_l-d_l.
\end{align*}
Swap $(i, j)$ and $(k, l)$, for $(k, l)(i, j)$ we have
\begin{align*}
L_k=2p+b_k-d_k\\
L_l=3p+b_l-d_l\\
L_i=5p+b_l+b_i-d_i\\
L_j=6p+b_l+b_j-d_j.
\end{align*}
Since $j>l$, we have $\ell_l\geq L_j$. Since $i>l$, we have $\ell_l\geq L_i$. Also $\ell_l\geq L_l$, and $\ell_k\geq L_k$. Thus $\max\{\ell_i,\ell_j,\ell_k,\ell_l\}\geq \max\{l_k,l_l\}\geq \max\{L_i,L_j,L_k,L_l\}$.

\end{proof}

The singleton jobs from $H_{\mathcal{I}}$ precede all pair of interlaced jobs.

\begin{lemma} \label{L14A}
Without loss of generality $(j, k)i$, where $i\in H_{\mathcal{I}}$ does not occur in $S$.
\end{lemma}
\begin{proof} Suppose it does. For $(j, k)i$ we have
\begin{align*}
\ell_j=2p+b_j-d_j\\
\ell_k=3p+b_k-d_k\\
\ell_i=5p+b_k+b_i-d_i.
\end{align*}
For $k(j, i)$
\begin{align*}
L_k=2p+b_k-d_k\\
L_j=4p+b_k+b_j-d_j\\
L_i=5p+b_k+b_i-d_i.
\end{align*}
Since $j\in \mathcal{I}\setminus H_{\mathcal{I}}$, we have $j>i$. Thus $L_i\geq L_j$. If $\ell_k>\ell_i$, then $\ell_k >\max\{L_i,L_j,L_k\}$. If $\ell_k\leq \ell_i$, then
$\ell_i=L_i=\max\{L_i,L_j,L_k\}$. Thus $\max\{L_i,L_j,L_k\}\leq \max\{\ell_i,\ell_j,\ell_k\}$.
%Since $j>l$, we have $l_l\geq L_j$. Also $l_l>L_l$. Since $j>l$ and $j\in H_{\mathcal{I}}$, we have $l\in H_{\mathcal{I}}$. By Lemma \ref{L12} $k>l$, thus
%$l_l\geq L_k$. Thus $\max\{l_j,l_k,l_l\}\geq l_l\geq \max\{L_j,L_k,L_l\}$.
\end{proof}

The following lemma summarizes the observations of this section.

\begin{lemma}
For  $\mathcal{I}$ with nonempty $H_{\mathcal{I}}=\{1, \dots, i\}$, $n\geq i\geq 1$ any optimal sequence $S$ starts with
\begin{equation}
1 \dots i-x(i+1, i-x+1)\dots(i+x, i) 
\end{equation}
where $1, \dots, i, \dots, i+x$, $0\leq x\leq i$,  are the $i+x$ jobs with the earliest due dates. That is
\begin{equation*}
S=1, \dots, i-x, (i+1, i-x+1)\dots(i+x, i)s,
\end{equation*}
where $s$ is the schedule for an instance $\mathcal{I}\setminus \{1, \dots, i, \dots, i+x\}$ without  long jobs.

\end{lemma}
\begin{proof} By Lemmas \ref{L14A} and \ref{L15} singletons $j\in H_{\mathcal{I}}$ make a prefix of $S$.  By Lemma \ref{L13} no pair $(i,j)$ with $j\in H_{\mathcal{I}}$ follows any pair $(k,l)$ with $j\in \mathcal{I}\setminus H_{\mathcal{I}}$.
By Lemma \ref{L12} the prefix of $S$  follows the EDD order.

\end{proof}

\section{Trimming algorithm for instances without long jobs}
This section gives trimming algorithm for the disagreeable case without long jobs. We further assume an even number of jobs $n$, the extension to an odd $n$ is, by Lemma \ref{LS}, straightforward and will be omitted. By Lemma \ref{LS} for an even $n$ the schedule has $n/2$ pairs of interlaced jobs.
We begin by defining pivotal jobs for given makespan $C_{\max}$ and maximum lateness $L_{\max}$.
\subsection{Pivotal jobs for $C_{\max}$ and $L_{\max}$} \label{S1}
 For an even $n$,  and given $C_{\max}\geq 3p \times \frac{n}{2}+b_n+ \dots b_{\frac{n}{2}+1}$ and $L_{\max}$ let 
 $k^*$ be the \emph{largest} $k=0, \dots, n-1$ such that
\begin{equation} \label{E1}
C_{\max}-d_{n-k}\leq L_{\max}.
\end{equation}
If no such $k$ exists, then there is no schedule with makespan $C_{\max}$ and maximum lateness not exceeding $L_{\max}$. 
If $k^*=n-1$, then the schedule
\begin{equation*}
(\frac{n}{2}, \frac{n}{2}+1) \dots, (1,n),
\end{equation*}
has maximum lateness not exceeding $L_{\max}$ and the makespan $ 3p \times \frac{n}{2}+b_n+ \dots b_{\frac{n}{2}+1}$. Thus we can assume that $k^*< n-1$.
Let $i^*$ be the \emph{largest} $n-1\geq i>k^*$ such that
\begin{equation} \label{E2}
C_{\max}-p-b_{n-k^*}+b_{n-i}-d_{n-i}\leq L_{\max}.
\end{equation}
If such $i^*$ does not exist, then the jobs in the last pair of  any schedule with makespan $C_{\max}$ and maximum lateness not exceeding $L_{\max}$ belong to the set
\begin{equation*}
\{n-k^*, \dots,n\}.
\end{equation*}
In this case, if $k^*=0$, then there is no schedule with makespan $C_{\max}$ and maximum lateness not exceeding $L_{\max}$. 
For $k^*>0$,  if there is a schedule with makespan $C_{\max}$ and maximum lateness not exceeding $L_{\max}$, then it ends with the pair $(n-k^*,n-k^*+1)$ (the pair is used for trimming in this case). Therefore it suffices to consider trimming for the case where both $k^*$ and $i^*$ exist and $n-1\geq i^*>k^*$.

\subsection{Possible ends}
\begin{lemma} \label{A12}
Let $S$ be the shortest schedule with maximum lateness $L_{\max}$, and let $C_{\max}$ be the makespan of $S$. If both $k^*$ and $i^*$ exist for $C_{\max}$ and $L_{\max}$, then $(n-i^*,w)$ for some $w> n-i^*$ occurs in $S$.
\end{lemma}

\begin{proof}
Suppose for contradiction that $(t,n-i^*)$ for some $t<n-i^*$ is in $S$ instead. We first show that $n-k^*$ does not occur \emph{before} that pair in $S$.  Suppose for contradiction that $n-k^*$ occurs \emph{before} $(t,n-i^*)$ in $S$, then by Lemma \ref{L11} we have
\begin{equation*}
S=\dots(t_0,n-k^*) \dots (t,n-i^*) \dots .
\end{equation*}
This leads to two cases, either $(t,n-i^*)$ is not last in $S$, i.e.
\begin{equation*}
S=...(t_0,n-k^*) \dots (t,n-i^*) \dots (t',w'),
\end{equation*}
or it is last, i.e.
\begin{equation*}
S=...(t_0,n-k^*) \dots (t,n-i^*).
\end{equation*}
In the latter case $C_{\max}-d_{n-i^*}\leq L_{\max}$.  Since $i^*>k^*$ this contradicts the definition of $k^*$. In the former, if $w'>n-k^*$, then swapping $n-k^*$ and $w'$ gives
\begin{equation*}
S'=..(t_0,w') \dots (t,n-i^*) \dots (t',n-k^*),
\end{equation*}
which does not increase maximum lateness and preserves the makespan. However $t'>n-i^*$ gives a shorter schedule by sweeping $t'$ with $n-i^*$ which leads to contradiction, and $t'<n-i^*$ contradicts the definition of $i^*$.
%Otherwise we get a shorter feasible
%schedule by swapping $n-k^*$ and $t'$ which gives contradiction. 
%Thus we must have $w'<n-k^*$ 
%Thus $w'<n-k^*$ which however contradicts the definition of $k^*$.
Thus $n-k^*$ must be \emph{after} $(t,n-i^*)$ in $S$. We now show that is assumption also leads to contradiction. We have that either $n-k^*$ is the first job interlaced with another job 
\begin{equation*}
S= \dots (t,n-i^*) \dots (n-k^*,w) \dots
\end{equation*}
or the second
\begin{equation*}
S= \dots (t,n-i^*) \dots (t', n-k^*) \dots
\end{equation*}
In the former case change $S$ to
\begin{equation*}
S'=\dots (t,w) \dots (n-i^*,n-k^*) \dots
\end{equation*}
which is shorter than $S$ which gives contradiction since maximum lateness does not increase.
In the latter case, if $(t', n-k^*)$ is last in $S$, i.e.
\begin{equation*}
S= \dots (t,n-i^*) \dots (t', n-k^*),
\end{equation*}
then for $t'>n-i^*$ swap  $t'$ and $n-i^*$ to obtain a shorter schedule which gives contradiction since maximum lateness does not increase. However $t'<n-i^*$ contradicts the definition of  $i^*$.  It remains to consider the case where $(t', n-k^*)$ is not last in $S$, i.e.
\begin{equation*}
S= \dots (t,n-i^*) \dots (t', n-k^*)\dots (t_0,w_0).
\end{equation*}
For $w_0<n-k^*$,  we get  contradiction with the definition of $k^*$. For $w_0>n-k^*$, 
swap of $n-k^*$ and $w_0$ to get
\begin{equation*}
S'=\dots (t,n-i^*) \dots (t',w_0)\dots (t_0, n-k^*)
\end{equation*}
without increasing maximum lateness.
If $t_0>n-i^*$, then the swap of $t_0$ and $n-i^*$ 
results in a shorter schedule without increasing maximum lateness
\begin{equation*}
S''= \dots (t,t_0) \dots (t', w_0)\dots (n-i^*,n-k^*)
\end{equation*}
which gives contradiction.
However $t_0<n-i^*$ contradicts the definition of  $i^*$. Thus $n-k^*$ cannot be after the pair $(t, n-i^*)$ either. This proves that  a pair $(n-i^*,w)$ for some $w> n-i^*$ exists in $S$.

\end{proof}

\begin{lemma} \label{A13}
Let $S$ be the shortest schedule with maximum lateness $L_{\max}$, and let $C_{\max}$ be the makespan of $S$. If both $k^*$ and $i^*$ exist for $C_{\max}$ and $L_{\max}$, then
$S$  ends with
\begin{equation*}
S_{\beta}=(n-i^*, n-i^*+2\beta+1)\dots (n-i^*+\beta, n-i^*+\beta+1)
\end{equation*}
for $\beta+1\geq i^*-k^*$, and
\begin{equation*}
S_{\beta}=(n-i^*, y_0)\dots (n-i^*+\beta, y_{\beta}=n-k^*),
\end{equation*}
for $\beta+1<i^*-k^*$, where $y_0>\dots>y_{\beta}$, and $0\leq \beta\leq \frac{i^*-1}{2}$. Moreover $\{y_0, \dots, y_{\beta}\}=\{n-i^*+\beta+1, \dots, n-i^*+2\beta+1\}$ for $2\beta+1\geq i^*-k^*$ 
and $\{y_0, \dots, y_{\beta-1}\}=\{n-i^*+\beta+1, \dots, n-i^*+2\beta\}$ for $2\beta+1< i^*-k^*$. 
\end{lemma}

\begin{proof} By Lemmas \ref{A12} and \ref{L11},  $S$ ends with the sequence
\begin{equation*}
(x_0=n-i^*, y_0)\dots (x_{\beta},y_{\beta})
\end{equation*}
for some $\beta\geq 0$, and
\begin{equation*}
x_0=n-i^*<\dots <x_{\beta}.
\end{equation*}
If $\beta=0$, then $S$ ends with
\begin{equation*}
(n-i^*, n-k^*).
\end{equation*}
If $\beta>0$, then $x_i=x_{i-1}+1$ for $i=1, \dots, \beta$. Suppose for contradiction that $x_{i-1}<z<x_i$ for some $i=1,\dots, \beta$. Then by Lemma \ref{L11} either
\begin{equation*}
S=\dots (t, z)\dots (x_{i-1},y_{i-1})(x_{i},y_{i}) \dots (x_{\beta},y_{\beta})
\end{equation*}
or
\begin{equation*}
S=\dots (x_{i-1},z)(x_{i},y_{i}) \dots (x_{\beta},y_{\beta}).
\end{equation*}
The swap of $z$ and $x_i$ gives either
\begin{equation*}
S'=\dots (t, x_{i})\dots (x_{i-1},y_{i-1})(z,y_{i}) \dots (x_{\beta},y_{\beta})
\end{equation*}
or
\begin{equation*}
S'=\dots (x_{i-1},x_{i})(z,y_{i}) \dots (x_{\beta},y_{\beta})
\end{equation*}
gives a shorter schedule since $z<x_i$, also the maximum lateness of the schedule does not exceed $L_{\max}$ since $n-i^*<z$. This results in contradiction.
Therefore we have
\begin{equation*}
S=\dots (n-i^*, y_0)\dots (n-i^*+\beta,y_{\beta}).
\end{equation*}
By definition of $k^*$ we have $y_{\beta}\geq n-k^*$. Thus $y_i<y_{i-1}$ for $i=1, \dots, \beta$ which follows by the exchange argument. Therefore
\begin{equation*}
(n-i^*, n-i^*+2\beta+1)\dots (n-i^*+\beta, n-i^*+\beta+1).
\end{equation*}
for $\beta+1\geq i^*-k^*$, and
\begin{equation*}
(n-i^*, y_0)\dots (n-i^*+\beta, y_{\beta}=n-k^*),
\end{equation*}
for $\beta+1<i^*-k^*$, where $\{y_0, \dots, y_{\beta}\}=\{n-i^*+\beta+1, \dots, n-i^*+2\beta+1\}$ for $2\beta+1\geq i^*-k^*$ 
and $\{y_0, \dots, y_{\beta-1}\}=\{n-i^*+\beta+1, \dots, n-i^*+2\beta\}$ for $2\beta+1< i^*-k^*$ as required by the lemma. For $S$ clearly $n-i^*+2\beta+1\leq n$.  Thus
$0\leq \beta\leq \frac{i^*-1}{2}$.
\end{proof}

\subsection{Optimal trim}
Let both $k^*$ and $i^*$ exist for $C_{\max}$ and $L_{\max}$. For  $\beta\geq 0$, $\alpha\geq 0$ such that $\alpha+2\beta+1=i^*$ and $i^*
+\alpha\leq n-1$ define \emph{trim} $\mathcal{S}_{\alpha,\beta}$ of length $\alpha+\beta+1$ as follows:

For $\beta+1\geq i^*-k^*$,  $\mathcal{S}_{\alpha,\beta}=P_{\alpha}S_{\beta}$ where
\begin{equation*}
P_{\alpha}=(n-i^*-\alpha, n) \dots(n-i^*-1,n-\alpha+1),
\end{equation*}
\begin{equation*}
S_{\beta}=(n-i^*, n-\alpha)\dots (n-i^*+\beta, n-\alpha - \beta).
\end{equation*}

For $\beta+1<i^*-k^*$, $\mathcal{S}_{\alpha,\beta}=P_{\alpha}S_{\beta}$ where
\begin{equation*}
P_{\alpha}=(n-i^*-\alpha, y_0) \dots(n-i^*-1,y_{\alpha -1}),
\end{equation*}
\begin{equation*}
S_{\beta}=(n-i^*, y_{\alpha})\dots (n-i^*+\beta, y_{\beta+\alpha}=n-k^*),
\end{equation*}
  $y_0>\dots>y_{\beta +\alpha-1}$, and $\{y_0, \dots, y_{\beta+\alpha}\}=\{n-i^*+\beta+1, \dots, n\}$. 

The trim $\mathcal{S}_{\alpha,\beta}$ is feasible if it has maximum lateness not exceeding $L_{\max}$ when completed at $C_{\max}$. Among feasible trims, if any, we select the one with maximum $\alpha$. The trim will be called
an optimal trim, $\mathcal{S}_{\alpha^*,\beta^*}=P_{\alpha^*}S_{\beta^*}$,  where $\alpha^*$ denotes an optimal $\alpha$, the optimal $\beta$ is denoted by $\beta^*=\frac{i^*-\alpha^*-1}{2}$, and optimal \emph{end} is  $S_{\beta^*}$. We later show that if no feasible trim exists, then no feasible schedule with maximum lateness not exceeding $L_{\max}$ and makespan $C_{\max}$ exists.

\begin{lemma} \label{A14}
Let $S$ be the shortest schedule with maximum lateness $L_{\max}$, and let $C_{\max}$ be the makespan of $S$. If both $k^*$ and $i^*$ exist for $C_{\max}$ and $L_{\max}$, then
$S$  ends with $S_{\beta\geq \beta^*}$.
\end{lemma}

\begin{proof} For contradiction suppose that $\beta<\beta^*$, i.e $\alpha>\alpha^*$ since $\alpha+2\beta+1=i^*$. We show that then a feasible trim with $\alpha$ exists which gives contradiction.
For $i^*-k^*\leq 2\beta +1$,  Lemma \ref{A13} implies that no job from the set  $A_{\beta}=\{ n-i^*+2\beta +2,\dots, n\}$ is in $S_{\beta}$
We have $|A_{\beta}|=i^*-2\beta-1=\alpha$. 
Also no job
from the set $B=\{1,\dots,n-i^*-1\}$ is in $S_{\beta}$. 
%Thus, if $2(i^*-\beta)<n$, then there is a pair $(x,y)$ with both $x$ and $y$ in $B$ in $S$.

For $i^*-k^*> 2\beta +1$  Lemma \ref{A13} implies that no job from the set  $A_{\beta}=\{ n-k^*+1,\dots, n\}\cup \{n-k^*-1, \dots, n-i^*+2\beta+1\}$ is in $S_{\beta}$,
Again, $|A_{\beta}|=i^*-2\beta-1=\alpha$. Also 
no job
from the set $B=\{1,\dots,n-i^*-1\}$ is in $S_{\beta}$. 
%Thus, if $2(i^*-\beta)<n$, then there is a pair $(x,y)$ with both $x$ and $y$ in $B$ in $S$. We assume that $(x,y)$ is the latest such pair in $S$.  
Thus $S$ looks as follows
\begin{equation*}
S=\dots (t_1,z_1)\dots(t_{\alpha}, z_{\alpha})\dots S_{\beta},
\end{equation*}
%for some $\lambda \geq 0$. For $\lambda=0$ we have $S=\dots (x,y)S_{\beta}$.  
where $\{z_1>\dots> z_{\alpha}\}= A_{\beta}$, and $\{t_1,\dots, t_{\alpha}\}\subseteq B$. By Lemma \ref{L11} we have
$t_1<\dots<t_{\lambda}<n-i^*$.  
%Thus $x\leq n-i^*-(\lambda+1)$. 
%Moreover since $n-i^*-(\lambda+1)\leq \min\{t_1,\dots, t_{\lambda}, y\}$ and $x<y$, we have $x< n-i^*-(\lambda+1)$. Thus $S$ looks as follows
%\begin{equation*}
%S=\dots (t,n-i^*-(\lambda+1))\dots (t_1,z_1)\dots(t_{\lambda},z_{\lambda})S_{\beta}.
%\end{equation*}
Let
\begin{equation*}
S_1=\dots (t_1,z_1)\dots (t_i,z_i)(x_1,y_1)\dots(x_{\ell},y_{\ell}),
\end{equation*}
be a feasible schedule in $[0,C]$ with maximum lateness not exceeding $L_{\max}$,
and
\begin{equation*}
S_2=(t_{i+1}, z_{i+1})\dots(t_{\alpha},z_{\alpha})S_{\beta}.
\end{equation*}
be a schedule in $[c\geq 0, C_{\max}]$ with maximum lateness not exceeding $L_{\max}$, and $c\leq C$.
Change $S_1$ to
\begin{equation*}
S'_1=\dots (t_1,z_1)\dots (t_i,x_1)(y_1,x_2)\dots(y_{\ell-1},x_{\ell})(y_{\ell},z_i),
\end{equation*}
observe that $x_i$ is longer than $y_i$ for the disagreeable case, and the jobs $x_1, y_1, \dots, x_{\ell},y_{\ell}$ reduce their completion times by at least $p$ in comparison to $S_1$. Thus $S'_1$ has maximum lateness
not exceeding $L_{\max}$. Consider $(x_i,y_i)$ from $S_1$ and $(y_i,x_{i+1})$ from $S'_1$ for $i=1, \dots, \ell_1$. For $x_i<y_i<x_{i+1}$, call  $(y_i,x_{i+1})$ in $S'_1$ a \emph{switch}. For $x_i<x_{i+1}<y_i$, change $(y_i,x_{i+1})$ to $(x_{i+1}, y_i)$ in $S'_1$
With no switches the makespan of $S'_1$ increases by $\Delta=b_{x_1}-b_{y_{\ell}}$ in comparison to $S_1$. Let  $(y_{i_1},x_{i_1+1})$, \dots,  $(y_{i_k},x_{i_k+1})$ where $k\geq 1$ and $i_1<\dots<i_k$ be all switches. The makespan of $S'_1$ increases  
\begin{equation*}
\Delta=(b_{x_1}-b_{y_{i_1}})+ (b_{x_{i_1+1}}-b_{y_{i_2}})+\dots+(b_{x_{i_{k-1}+1}}-b_{y_{i_k}})+ (b_{x_{i_{k}+1}}-b_{y_{\ell}})
\end{equation*}
in comparison to $S_1$.
Change $S'_1$ to
\begin{equation*}
S''_1=\dots (t_1,z_1)\dots (t_i,x_1)(y_1,x_2)\dots(y_{\ell-1},x_{\ell})
\end{equation*}
in $[0, C+\Delta-(3p+b_{z_i})]$. The schedule is feasible and its maximum lateness do not exceed $L_{\max}$. 
We change $S_2$ to a feasible schedule with maximum lateness
not exceeding $L_{\max}$

\begin{equation*}
S'_2=(y_{\ell},z_i),(t_{i+1}, z_{i+1})\dots(t_{\alpha},z_{\alpha})S_{\beta}
\end{equation*}
in $[c-(3p+b_{z_i}), C_{\max}]$. We have
\begin{equation*}
C+\Delta-(3p+b_{z_i})\geq c-(3p+b_{z_i})
\end{equation*}
since $C\geq c$.

By repeating this operation we get 
\begin{equation*}
S_2=(t_1,z_1)\dots(t_{\alpha}, z_{\alpha})S_{\beta}
\end{equation*}
with maximum lateness not exceeding $L_{\max}$ when completed at $C_{\max}$. Also we have
\begin{equation*}
t_{k}\leq n-i^*+k -\alpha-1
\end{equation*}
for $k=1, \dots, \alpha$. Thus
\begin{equation*}
(n-i^*-\alpha,z_1)\dots(n-i^*-1, z_{\alpha})S_{\beta}
\end{equation*}
is a feasible trim $P_{\alpha}S_{\beta}$ with maximum lateness not exceeding $L_{\max}$ when completed at $C_{\max}$.
Since $\alpha>\alpha^*$ we get contradiction.
\end{proof}

\begin{lemma} 
Let $S$ be the shortest schedule with maximum lateness $L_{\max}$, and let $C_{\max}$ be the makespan of $S$. If both $k^*$ and $i^*$ exist for $C_{\max}$ and $L_{\max}$, then
$S$  ends with $S_{\beta^*}$.
%There is an optimal schedule that ends with the suffix $S_{\beta^*}$.
\end{lemma}

\begin{proof}  
Let $S$ be the shortest schedule with maximum lateness $L_{\max}$. 
%Without loss of generality we
%assume that $S$ is the shortest schedule with maximum lateness $L_{\max}$. 
Let $C_{\max}$ be the makespan of $S$. By Lemma \ref{A14} let the $S$ ends with 
\begin{equation*}
S_{\beta\geq \beta^*}.
\end{equation*}
Suppose for contradiction $\beta^*<\beta$.
We show that the larger $\beta$ guarantees a schedule shorter than $S$ with maximum lateness not exceeding $L_{\max}$ which gives contradiction.
For $i^*-k^*\leq 2\beta +1$,  Lemma \ref{A13} implies that no job from the set  $A_{\beta}=\{ n-i^*+2\beta +2,\dots, n\}$ is in $S_{\beta}$. We have $|A_{\beta}|=i^*-2\beta-1$. Also no job
from the set $B=\{1,\dots,n-i^*-1\}$ is in $S_{\beta}$. Thus, if $2(i^*-\beta)<n$, then there is a pair $(x,y)$ in $S$ with both $x$ and $y$ in $B$.

Similarly, for $i^*-k^*> 2\beta +1$  Lemma \ref{A13} implies that no job from the set  $A_{\beta}=\{ n-k^*+1,\dots, n\}\cup \{n-k^*-1, \dots, n-i^*+2\beta+1\}$ is in $S_{\beta}$. Again, $|A_{\beta}|=i^*-2\beta-1$. Also no job
from the set $B=\{1,\dots,n-i^*-1\}$ is in $S_{\beta}$. Thus, if $2(i^*-\beta)<n$, then there is a pair $(x,y)$ with both $x$ and $y$ in $B$ in $S$. We assume that $(x,y)$ is the latest such pair in $S$.  
Thus $S$ looks as follows
\begin{equation*}
S=\dots (x,y)(t_1,z_1)\dots(t_{\omega},z_{\omega})S_{\beta}
\end{equation*}
for some $i^*-2\beta-1\geq \omega \geq 0$. For $\omega=0$ we have $S=\dots (x,y)S_{\beta}$.  Moreover $\{z_1>\dots> z_{\omega}\}\subseteq A_{\beta}$  and the jobs are the $\omega$ longest jobs in $A_{\beta}$, and $\{t_1,\dots, t_{\omega}\}\subseteq B$. By Lemma \ref{L11} we have
$x<t_1<\dots<t_{\omega}<n-i^*$.  Thus $x\leq n-i^*-(\omega+1)$. Moreover since $n-i^*-(\omega+1)\leq \min\{t_1,\dots, t_{\omega}, y\}$ and $x<y$, we have $x< n-i^*-(\omega+1)$. Thus $S$ looks as follows
\begin{equation*}
S=\dots (t,n-i^*-(\omega+1))\dots (t_1,z_1)\dots(t_{\omega},z_{\omega})S_{\beta},
\end{equation*}
where both $t$ and $n-i^*-(\omega+1)$ are in $B$. We change $S$ by shortening $S_{\beta}$ and swapping $n-i^*-(\omega+1)$ to make the schedule shorter.
For $\beta+1> i^*-k^*$
\begin{equation*}
S_{\beta}=(n-i^*, n-i^*+2\beta+1)\dots (n-i^*+\beta, n-i^*+\beta+1)
\end{equation*}
change to
\begin{equation*}
S_{\beta-1}=(n-i^*, n-i^*+2\beta-1)\dots (n-i^*+\beta-1, n-i^*+\beta),
\end{equation*}
the jobs $w=n-i^*+2\beta$ and $v=n-i^*+2\beta+1$ are pushed out of $S_{\beta}$.
For $\beta+1= i^*-k^*$
\begin{equation*}
S_{\beta}=(n-i^*, n-i^*+2\beta+1)\dots (n-i^*+\beta, n-i^*+\beta+1)
\end{equation*}
change to
\begin{equation*}
S_{\beta-1}=(n-i^*, n-i^*+2\beta-1)\dots (n-i^*+\beta-1, n-i^*+\beta+1),
\end{equation*}
the jobs $w=n-i^*+2\beta$ and $v=n-i^*+2\beta+1$ are pushed out of $S_{\beta}$.
For $\beta+1<i^*-k^*$
\begin{equation*}
S_{\beta}=(n-i^*, y_0)\dots (n-i^*+\beta, y_{\beta}=n-k^*),
\end{equation*}
where $y_0>\dots>y_{\beta}$, and $\{y_0, \dots, y_{\beta}\}=\{n-i^*+2\beta+1, \dots, n-i^*+\beta+1\}$ for $2\beta+1\geq i^*-k^*$ 
and $\{y_0, \dots, y_{\beta-1}\}=\{n-i^*+2\beta, \dots, n-i^*+\beta+1\}$ for $2\beta+1< i^*-k^*$.  Change to
\begin{equation*}
S_{\beta-1}=(n-i^*, y_2)\dots (n-i^*+\beta-1, y_{\beta}=n-k^*),
\end{equation*}
the jobs $v=y_0$ and $w=y_1$ are pushed out of $S_{\beta}$. We now modify $S$ as follows.
For $\omega\geq 2$ change $S$ to
\begin{equation} \label{X1}
S'=\dots (t,z_1)\dots (n-i^*-(\omega+1),z_2)(t_1,z_3)\dots(t_{\omega-2},z_{\omega})(t_{\omega-1}, v)(t_{\omega}, w)S_{\beta-1}.
\end{equation}
For $\omega=1$ change $S$ to
\begin{equation} \label{X2}
S'=\dots (t,z_1)\dots (n-i^*-(\omega+1),v)(t_1,w)S_{\beta-1}
\end{equation}
respectively.
For $\omega=0$ change $S$ to
\begin{equation} \label{X3}
S'=\dots (t,v)\dots (n-i^*-(\omega+1),w)S_{\beta-1}.
\end{equation}
It can be checked that the schedules (\ref{X1})-(\ref{X3}) have maximum lateness not exceeding $L_{\max}$ and they are shorter than $C_{\max}$ since $n-i^*-(\omega+1)$ was swapped with a shorter job
which gives contradiction.
\end{proof}

\subsection{The trimming algorithm}

Let $\mathcal{I}=\{1, \dots, n\}$ be an instance with an even $n$. Let $C_{\mathcal{I}}=3p\frac{n}{2}+b_{\frac{n}{2}+1}+\dots+b_n$ be the lower bound on the makespan in the disagreeable case. Do \emph{binary} search on $C_{\max}$ in the interval $[3p\frac{n}{2}+b_{\frac{n}{2}+1}+\dots+b_n, 3p\frac{n}{2}+b_1+ \dots+ b_{\frac{n}{2}}]$, and $L_{\max}$ in the interval $[3p\frac{n}{2}+b_{\frac{n}{2}+1}+\dots+b_n -d_n, 3p\frac{n}{2}+b_1+ \dots+ b_{\frac{n}{2}}-d_1]$. For given $C_{\max}\geq C_{\mathcal{I}}$ and $L_{\max}$ test whether there is a schedule $S$ for  $\mathcal{I}$ with maximum lateness not exceeding $L_{max}$, and makespan and not exceeding $C_{\max}$. 
The test works as follows. 

\emph{Initial step}: Set $\mathcal{I}:=\{1, \dots, n\}$, $C:=C_{\max}$, and an empty schedule $S$. 

\emph{Main step}: For  $\mathcal{I}$, $C\geq C_{\mathcal{I}}$ and $L$ calculate the pivotal jobs $k^*$ and $i^*$. If $k^*$ does not exist, then no schedule with maximum lateness not exceeding $L_{max}$ for $\{1, \dots,n\}$ exists. If $i^*$ does not exist, then set $S:=S(n-k^*,n-k^*+1)$ (if $i^*$ does not exist, then $k^*>0$, otherwise no schedule with maximum lateness not exceeding $L_{max}$ for $\{1, \dots,n\}$ exists, see the beginning of Subsection \ref{S1}). Set $\mathcal{I}:=\mathcal{I}\setminus \{n-k^*,n-k^*+1\}$ (the trimming)
and $C:=C-(3p+b_{n-k^*+1})$. If $C\geq C_{\mathcal{I}}$ and $\mathcal{I}\neq \emptyset$, then repeat the main step. If $C<C_{\mathcal{I}}$, then no schedule with maximum lateness not exceeding $L_{max}$ for $\{1, \dots,n\}$ exists, stop. If $C\geq 0$ and $\mathcal{I}=\emptyset$, then the schedule $S$ has maximum lateness not exceeding $L_{max}$, stop.

If $i^*$ exists, then let $\mathcal{S}_{\beta^*}$ be an optimal end (if the end does not exist, then by Lemma \ref{A13} and the proof of Lemma \ref{A14} no schedule with maximum lateness not exceeding $L_{max}$ for $\{1, \dots,n\}$ exists). Set $S:=S\mathcal{S}_{\beta^*}$. Remove the jobs $I_{\beta^*}$ from $\mathcal{S}_{\beta^*}$ from the instance i.e. set $\mathcal{I}:=\mathcal{I}\setminus I_{\beta^*}$ (the trimming),
and reduce $C$ by the makespan $C_{\beta^*}$ of $\mathcal{S}_{\beta^*}$, i.e. set $C:= C-C_{\beta^*}$. If $C\geq C_{\mathcal{I}}$ and $\mathcal{I}\neq \emptyset$, then repeat the main step. If $C<C_{\mathcal{I}}$, then no schedule with maximum lateness not exceeding $L_{max}$ for $\{1, \dots,n\}$ exists, stop. If $C\geq 0$ and $\mathcal{I}=\emptyset$, then the schedule $S$ has maximum lateness not exceeding $L_{max}$, stop.
%Repeat the step until there is no job left to schedule.
%\end{proof}

%We have now two cases to consider. In either case we show that there always is an optimal  schedule with a fixed  \emph{suffix}.
% which depends on whether $i^*-k^*>1$ or $i^*-k^*=1$. 
% The suffix can then be 
%fixed and the algorithm repeated with fewer jobs and reduced makespan. 

\section{General instances. An algorithm  for given partition  $S$ and $T$} \label{G}
\label{sec:3}

In this section we show how to test whether there is a schedule with maximum lateness not exceeding $L$ for given partition $P=\{1, \dots,m\}$ and $T=\{1', \dots, m'\}$ of $\mathcal{J}$ of size $m\geq '$, $m+m'=n$. We have $b_{j'}\leq p$ for each $j'\in T$. We assume without loss of generality that $m'>0$ since for $m'=0$ the singletons in $P$ follow the EDD order. We have $1'\prec 2' \prec \dots \prec m'$ for the jobs in $T$.
% with $b_j\leq p$, and $B_j= b'_j+3p$ for job $J_j\in \mathcal{J}$ with $b_j> p$, where $b'_j=b_j-p$.  Given partition into $S$ and $T$ . 
Let 
%We need to add a key condition
%\begin{equation}
%\beta - (B_{\pi(n)}+ \dots +B_{\pi(1)})\geq 0.
%\end{equation}
%The condition is equivalent to
\begin{equation}
\beta =3pm'+2p(m-m')+\sum_{j\in P}b_{j}.
\end{equation}
%which is easy to check for a given partition.
%Let $D_1, \dots, D_n$ be due dates, and $c_1, \dots, c_n$ processing times respectively of the jobs  in $T$. 
%that the jobs in $T$ are order so that  $D_n-c_n\geq \dots \geq D_1-c_1$. 
Set $i:=0$, and $P_m:=P$.
Let $\pi(m)\in P_{m}$ be the \emph{longest} job that meets the  inequality

\begin{equation*} 
\beta  \leq L+d_{\pi(m)}.
\end{equation*}
If no job in  $P_{m}$ satisfies the inequality, then answer No and stop. Otherwise, set $P_{m-1}:=P\setminus \{\pi(m)\}$ and check if job $m'$ satisfies

\begin{equation*} 
\beta - b_{\pi(m)}  \leq L+d_{m'}-p-b_{m'},
\end{equation*}
if it does, then set $\beta:=\beta - (3p+b_{\pi(m)})$, $T:=T\setminus \{m'\}$, and $m':=m'-1$. If it does not, then set $\beta:=\beta - (2p+b_{\pi(m)})$. Continue.
%If no job in  $S_{n}$ satisfies the inequality, then answer No and stop. Otherwise, continue.

Suppose that partial permutation $\pi(m), \dots, \pi(m-i+1)$ of the jobs in $P$ has been calculated for some $0< i \leq m-1$. Let $P_{m-i}= P \setminus \{ \pi(m), \dots, \pi(m-i+1) \}$. Let $\pi(m-i)\in P_{m-i}$ be the \emph{longest} job that meets
the inequality 
\begin{equation*} 
\beta  \leq L+d_{\pi(m-i)}.
\end{equation*}
If no job in  $P_{m-i}$ satisfies the inequality, then answer No and stop. Otherwise, set $P_{m-(i+1)}:=P\setminus \{ \pi(m), \dots, \pi(m-i) \}$. If $m'=0$, then set $\beta:=\beta - (2p+b_{\pi(m)})$. Continue.
Otherwise check if job $m'$ satisfies 
\begin{equation*} 
\beta - b_{\pi(n-i)}  \leq L+d_{n'}-p-b_{n'}
\end{equation*}
if it does, then set $\beta:=\beta - (3p+b_{\pi(m-i)})$, $T:=T\setminus \{m'\}$, and $m':=m'-1$. If it does not, then set $\beta:=\beta - (2p+b_{\pi(m-i)})$. Continue.
%set $\pi(n-i):=J_j$. 

Suppose that partial permutation $\pi(m), \dots, \pi(m-i+1)$ of the jobs in $S$ has been calculated for $m=i$. If $m'=0$, then answer Yes and stop. Otherwise, if $m'>0$, then answer No and stop.

We show that any schedule $\sigma$ for partition $P$ and $T$ with maximum lateness not exceeding $L$ can be turned into the schedule obtained by the test without increasing maximum lateness. By exchange argument the permutation of the jobs in $P$ defined by $\sigma$ can be changed so that the job $\sigma(m)$ is the longest among all jobs in $P$
that have maximum lateness not exceeding $L$ when completed at $\beta$. Thus $\pi(m)=\sigma(m)$ in the schedule. Similar arguments show $\sigma(m-1)=\pi(m-1), \dots, \sigma(1)=\pi(1)$ in the schedule. Furthermore, let jobs in $T$ be paired as follows
$(1', \pi(i_1)), \dots, (m', \pi(i_{m'}))$ with the jobs $\pi(i_1), \dots,  \pi(i_{m'})$ from $P$ in positions $i_1, \dots, i_{m'}$ of $\sigma$. The schedule $\sigma$ can be changed so that the pair of $m'$ with any job from $P$ in positions $j>i_{m'}$
results in the maximum lateness of $m'$ exceeding $L$. Similar arguments for the remaining jobs in $T$ show that the resulting schedule is the same as the one obtained by the test.

\section{Conclusion}

We studied the problem $1|(p,p,b_i)|L_{\max}$ of minimization of maximum lateness for coupled tasks with exact delays scheduled on a single machine. Surprisingly the optimization algorithms for the problem turn out to be rather involved even for
special cases of agreeable and disagreeable orders of due dates and processing times $b_i$. Nevertheless we show  polynomial time algorithms for both cases using quite intricate arguments. The computational complexity of the general problem
remains open though we show that an optimal schedule can be found as long as we know which jobs determine schedule makespan. We conjecture that a pseudopolynomial algorithm exists for the general problem.

\bibliographystyle{plain}
\bibliography{BibChapter1a}
\end{document}